\documentclass[11pt]{article}
\usepackage{authblk}
\usepackage{amssymb,latexsym,amsmath}
\usepackage{amsthm}

\theoremstyle{definition}
\newtheorem{theorem}{Theorem}
\newtheorem{corollary}[theorem]{Corollary}
\newtheorem{proposition}[theorem]{Proposition}
\newtheorem{lemma}[theorem]{Lemma}
\newtheorem{definition}[theorem]{Definition}
\newtheorem{example}[theorem]{Example}
\newtheorem{notation}[theorem]{Notation}
\newtheorem{remark}[theorem]{Remark}

\usepackage{graphicx}
\usepackage{tikz}
\usetikzlibrary{matrix,decorations.pathreplacing}

\usepackage{relsize}

\newcommand{\numberset}{\mathbb}

\newcommand{\N}{\numberset{N}}

\newcommand{\F}{\numberset{F}}

\newcommand{\mC}{\mathcal{C}}
\newcommand{\mD}{\mathcal{D}}

\newcommand{\mG}{\mathcal{G}}

\usepackage{hyperref}
\newcommand{\mA}{\mathcal{A}}

\usepackage[margin=2.8cm]{geometry}

\begin{document}

\title{\textbf{Rank-metric codes  and their duality theory}}

\author{Alberto Ravagnani%
  \thanks{E-mail: \texttt{alberto.ravagnani@unine.ch}. The author
was partially supported by the Swiss National Science
Foundation through grant no. 200021\_150207.}}
\affil{Institut de Math\'{e}matiques \\ Universit\'{e} de
Neuch\^{a}tel\\Emile-Argand 11, CH-2000 Neuch\^{a}tel, Switzerland}

\date{}

\maketitle

\begin{abstract}
We compare the two duality theories of rank-metric codes proposed by Delsarte
and Gabidulin, 
proving that the former generalizes the latter.
We also 
give an elementary proof of MacWilliams identities for the general case of
Delsarte rank-metric
codes. The identities which we derive are very easy to handle, and allow
us to re-establish in a very
concise way the main results of the theory of rank-metric codes
first proved by Delsarte employing the theory of association schemes and regular
semilattices.
We also show that our identities imply  as a corollary the original MacWilliams
identities established by Delsarte.
We describe how the minimum and maximum rank of a
rank-metric code relate to the minimum and maximum rank of the dual
code, 
giving some bounds
and characterizing the codes attaining them. Then we study optimal
anticodes in the
rank metric, describing them in terms of optimal codes (namely, MRD codes). In particular, we prove
that the dual of an optimal anticode is
an
optimal anticode. Finally, as an application of our results to a classical
problem in enumerative
combinatorics, we derive both a recursive and an explicit  formula for the number
of $k \times m$ matrices over a finite field with given rank and $h$-trace.
\end{abstract}

\section*{Introduction}\label{intr0}

In \cite{del1} Delsarte defines rank-metric codes as sets of matrices of given
size over a finite field $\F_q$.
The distance between two matrices is given by
the rank of their difference. Interpreting matrices as bilinear forms,
Delsarte studies rank-metric codes as  association schemes, whose adjacency
algebra yields the so-called MacWilliams
transform of distance enumerators of codes. The results of \cite{del1} are based on the general 
theory of designs and codesigns in regular 
semilattices
developed in \cite{del2}.

In \cite{gabid} Gabidulin proposed a different definition of rank-metric code,
in which the
codewords are vectors with entries in an extension field $\F_{q^m}$. 
MacWilliams identities for
Gabidulin codes were obtained in \cite{gadu}. As we will explain in details,
one can naturally  associate to any
vector with entries in an extension field $\F_{q^m}$ a matrix of prescribed size
over the base field $\F_q$. The rank of the
vector (as defined by
Gabidulin) coincides with the ordinary rank of the associated
matrix.
Hence there exists a natural way
to associate to a Gabidulin
 code a Delsarte code with the same metric
properties. From this point of view, Gabidulin codes
can be regarded as a special case of Delsarte codes. It
is however not clear  in general how the duality theories
of these two families of codes relate to each other. This is one of the
questions that we address in this work.

Both linear Delsarte and Gabidulin codes have interesting applications in
information theory. 
Recently it was shown how to employ
them for error correction in non-coherent linear network coding and in coherent 
linear network coding under an adversarial channel model
 (see e.g. \cite{metrics} and the 
references within). 
Rank-metric codes were also proposed to secure a network coding
communication 
system against an eavesdropper in a
universal way (see \cite{metrics3} for details).

Motivated by these
applications, in this paper we study the duality 
theories of linear Delsarte
and Gabidulin codes, mainly focusing on their MacWilliams identities.
In 
coding theory, a MacWilliams identity establishes a relation between metric
properties of a 
code and
metric properties of the dual code. MacWilliams identities exist for
several types of codes and metrics.
As Gluesing-Luerssen observed in  \cite{heide1}, association schemes  provide
the most general approach to 
MacWilliams identities, and apply to both linear and non-linear codes (see
\cite{del3}, 
\cite{camion} and \cite{delslev}).
On the other hand, the machinery of association schemes and of the related
Bose-Mesner algebras 
is a very elaborated mathematical tool.
Several authors proved independently  the MacWilliams identities for the
various
types of codes in less sophisticated ways.

A different viewpoint on MacWilliams identities for general additive codes was
recently 
proposed by Gluesing-Luerssen in \cite{heide1}. The
approach is based on character 
theory and partitions of groups. See also \cite{grant} for
a  character-theoretic approach to MacWilliams identities for the rank
and the Hamming metric.

Both the theory of association schemes and the approach of \cite{heide1}  apply
to Delsarte rank-metric codes, giving 
MacWilliams identities in the form 
presented in \cite{del1}. On the other hand, to the extent of our knowledge,
there is no 
elementary proof for them.

Let us briefly describe the content of the paper. We start comparing the
duality theories of Delsarte and Gabidulin codes, proving that the former
generalize the latter.  
Then we give a  short proof
of MacWilliams identities for the general case of Delsarte rank-metric codes.
We
only employ elementary
properties of the rank metric, linear algebra techniques and a double counting
argument, avoiding
 any sort of numerical calculation.
The identities which we derive have a very convenient form, which allows us to
re-establish  the most important results of the theory of rank-metric codes in
a concise way. We also show that the original MacWilliams
identities proved by Delsarte in \cite{del1} can be easily obtained
from our identities.
In a second part we  prove some  bounds that relate the minimum and maximum rank of a
code to the minimum and maximum rank of the dual code, characterizing the codes
which attain them. The bounds show that also the maximum rank of a code
(and not only the minimum rank) deserves interest. We also investigate
anticodes in the rank metric, and present a new
characterization of optimal anticodes in terms of
optimal codes. Then we apply such characterization to show that
the dual of
an optimal anticode is an optimal anticode. This result may be regarded as the
analogue for anticodes of the fact that the dual of an optimal code is an
optimal code.
Finally, as an application of
our results to a classical problem in enumerative combinatorics, we give
both a recursive and an explicit formula for the
number of rectangular matrices with given rank and $h$-trace over a finite
field $\F_q$. To the extent of
our knowledge, formulas of this type are not available in the literature.

The layout of the paper is as follows. Section \ref{sec1} contains some
preliminaries on  Delsarte
rank-metric codes.
In Section \ref{sec2} we compare Delsarte and Gabidulin codes. In Section \ref{sec3} 
we give an elementary proof
for  MacWilliams identities for the general case of Delsarte  rank-metric
codes, and use them to establish the main 
results of the  theory of these codes. In Section
\ref{sec4} and in the last part of Section \ref{sec4a} we study how the minimum and 
the maximum rank  of a rank-metric code
relate to the minimum and maximum rank of the dual code, proving some
bounds on the 
involved parameters and characterizing the codes
which attain them. Optimal anticodes in the rank metric are studied in Section
\ref{sec4a}. In Section \ref{sec5} we derive a recursive formula
for the number of rectangular matrices over $\F_q$ of given rank and 
$h$-trace. In Appendix \ref{appa} we show how our results relate to the original work by
Delsarte, giving in particular an explicit formula for the number of rectangular 
matrices over $\F_q$ of given rank and 
$h$-trace.

\section{Preliminaries} \label{sec1}

Throughout this paper, $q$ denotes a fixed prime power, and $\F_q$ the finite field with $q$ elements. We
also work with positive integers $k$ and $m$.
\begin{notation}
 We denote by $\mbox{Mat}(k\times m, \F_q)$ the $\F_q$-vector space of $k
\times m$ 
matrices with
entries in $\F_q$. Given a matrix $M \in \mbox{Mat}(k\times m, \F_q)$, we write 
$\mbox{Tr}(M)$ for the trace of $M$, and $M_i$ for the $i$-th column of $M$,
i.e., the vector $(M_{1i},M_{2i},...,M_{ki})^t \in \F_q^k$. 
The transpose of $M$ is $M^t$, while $\mbox{rk}(M)$ denotes the rank of $M$. 
The vector space generated by
the columns of a matrix $M \in \mbox{Mat}(k\times m, \F_q)$ is
$\mbox{colsp}(M) \subseteq \F_q^k$. When the size is clear from the context, we denote a zero matrix simply by $0$.
\end{notation}

Let us briefly recall the setup of \cite{del1}.

\begin{definition}
 The \textbf{trace product} of matrices $M,N \in \mbox{Mat}(k\times
m, \F_q)$ is denoted and defined by
 $$\langle M , N \rangle := \mbox{Tr}(MN^t).$$
\end{definition}

 It is  easy to check that the map
$\langle \cdot , \cdot \rangle: \mbox{Mat}(k \times m, \F_q) \times \mbox{Mat}(k \times m, \F_q) \to \F_q$
is a scalar product (i.e., symmetric, bilinear and non-degenerate).

\begin{definition} \label{variedefgab}
 A (\textbf{Delsarte rank-metric}) \textbf{code} of size $k\times m$ over
$\F_q$
is an $\F_q$-linear subspace 
 $\mC \subseteq \mbox{Mat}(k \times m, \F_q)$.
  The \textbf{minimum rank} of a non-zero code $\mC$ is denoted and defined by
 $\mbox{minrk}(\mC):= \min \{ \mbox{rk}(M) : M \in \mC, \ \mbox{rk}(M) >0\}$,
while the \textbf{maximum rank} of any code $\mC$
 is denoted and defined by $\mbox{maxrk}(\mC):= \max \{ \mbox{rk}(M) : M \in \mC\}$.
 The \textbf{dual} of $\mC$ is  the Delsarte code
 $\mC^\perp:= \{ N \in \mbox{Mat}(k \times m, \F_q) : \langle M, N \rangle =0 \mbox{ for all } M \in \mC\}$.
\end{definition}

\begin{remark}
  We notice that Delsarte codes are by definition linear over $\F_q$, and
that linearity is a crucial property for the results presented in this paper.
\end{remark}

The following lemma summarizes some straightforward properties of duality.

\begin{lemma}\label{proprduale}
 Let $\mC,\mD \subseteq \mbox{Mat}(k \times m, \F_q)$ be codes. We have:
 \begin{itemize}
  \item $(\mC^\perp)^\perp=\mC$;
  \item $\dim_{\F_q}(\mC^\perp)= km-\dim_{\F_q}(\mC)$;
  \item $(\mC \cap \mD)^\perp=\mC^\perp + \mD^\perp$, and $(\mC +
\mD)^\perp=\mC^\perp \cap \mD^\perp$.
 \end{itemize}
\end{lemma}

Recall that for $n \in \N_{\ge 1}$ the standard inner product of vectors 
$(x_1,...,x_n),(y_1,...,y_n) \in \F_q^n$ is defined by $(x_1,...,x_n) \cdot
(y_1,...,y_n):= \sum_{i=1}^n x_iy_i$. It is easy to see that the trace product
$\langle \cdot, \cdot \rangle$ on $\mbox{Mat}(k\times m,
\F_q)$ and the standard inner product on $\F_q^k$ relate as follows.

\begin{lemma} \label{rel}
 Let $M, N \in \mbox{Mat}(k\times m, \F_q)$. We have $\langle M, N \rangle=
\sum_{i=1}^m M_i \cdot N_i$.
\end{lemma}

\begin{notation}
 Lemma \ref{rel} says that the trace product of two matrices is the sum of the 
 standard inner products of the columns of the two matrices. In
particular, the trace product $\langle \cdot,
\cdot \rangle$ on 
 $\mbox{Mat}(k \times 1, \F_q) \cong \F_q^k$ coincides with the standard inner 
 product on $\F_q^k$. Hence from now on we denote both products by $\langle
\cdot , \cdot \rangle$.
 We also denote by $U^\perp$ the dual of a vector subspace $U \subseteq
\F_q^k$, i.e., $U^\perp:= \{ x \in \F_q^k : \langle u,x\rangle =0 \mbox{
for all } u \in U\}$.
\end{notation}

The following result, first proved by Delsarte, is well-known.
\begin{theorem}[\cite{del1}, Theorem 5.4] \label{bound1}
 Let $\mC \subseteq \mbox{Mat}(k \times m, \F_q)$ be a non-zero code, and let $d:=\mbox{minrk}(\mC)$. We have
 $$\dim_{\F_q} (\mC) \le \max \{ k,m\} (\min \{ k,m\}-d+1).$$ Moreover, for any
$1 \le d \le \min \{k,m \}$ there exists a non-zero code 
 $\mC \subseteq \mbox{Mat}(k \times m, \F_q)$ of minimum rank $d$ which attains the upper bound.
\end{theorem}
\begin{definition} \label{defmrd}
 A code attaining the bound of Theorem \ref{bound1} is said to be an \textbf{optimal} or
\textbf{maximum rank distance} code (\textbf{MRD} code in short).
 The zero code will be also considered MRD.
\end{definition}

\begin{remark} \label{tuttomrd}
  Notice that $\mbox{Mat}(k \times m, \F_q)$ is a trivial example of an MRD code
with minimum rank $1$ and dimension $km$.
 See \cite{del1}, Section 6, for the construction of codes attaining the bound of Theorem \ref{bound1} for any choice of the parameters.
\end{remark}

\begin{definition}
 Given a code $\mC$ and an integer $i \in
\N_{\ge 0}$ define
 $A_i(\mC):= |\{ M \in \mC : \mbox{rk}(M)=i\}|$. The collection
${(A_i(\mC))}_{i\in \N_{\ge 0}}$ is said to be the \textbf{rank distribution}
 of $\mC$.
\end{definition}

\begin{remark}
 The minimum rank of a non-zero code $\mC \subseteq \mbox{Mat}(k \times m, \F_q)$
is the smallest $i>0$ such that $A_i(\mC)>0$.
 Notice that we define $A_i(\mC)$ for any $i \in \N_{\ge 0}$, even if we clearly have $A_i(\mC)=0$ for all integers $i > \min \{ k,m\}$. 
 This choice will
 simplify the statements in the sequel.
\end{remark}

\section{Delsarte and Gabidulin rank-metric codes} \label{sec2}
A different definition of rank-metric code, proposed by Gabidulin, is
the following.
\begin{definition}[see \cite{gabid}]
 Let $\F_{q^m}/\F_q$ be a finite field extension. A \textbf{Gabidulin}
(\textbf{rank-metric}) \textbf{code}
 of length $k$ over $\F_{q^m}$ is an $\F_{q^m}$-linear subspace $C \subseteq
\F_{q^m}^k$. 
 The \textbf{rank} of a vector $\alpha=(\alpha_1,...,\alpha_k) \in \F_{q^m}^k$ is defined as
 $\mbox{rk}(\alpha):= \dim_{\F_q} \mbox{Span} \{ \alpha_1,...,\alpha_k\}$.
 The \textbf{minimum rank} of a Gabidulin code $C \neq 0$ is
$\mbox{minrk}(C):= \min \{ \mbox{rk}(\alpha):
 \alpha \in C, \ \alpha \neq 0\}$, and the \textbf{maximum rank} of any
 Gabidulin code $C$ is $\mbox{maxrk}(C):= \max \{ \mbox{rk}(\alpha) : \alpha \in C\}$.
 The  \textbf{rank distribution} of $C$ is the collection
 ${(A_i(C))}_{i \in \N_{\ge 0}}$, where $A_i(C):= |\{ \alpha \in C :
\mbox{rk}(\alpha) =i\}|$.
 The \textbf{dual} of a Gabidulin code $C$ is denoted and defined by
$C^\perp:=\{ \beta \in \F_{q^m}^k :
\langle \alpha, \beta \rangle =0\mbox{ for all }
 \alpha \in C\}$, where $\langle \cdot , \cdot \rangle$ is the standard
inner product of $\F_{q^m}^k$.
\end{definition}

It is natural to ask how Gabidulin
and Delsarte codes relate to each other.

\begin{definition}\label{assoc}
 Let $\mG= \{ \gamma_1,...,\gamma_m\}$ be a basis
of $\F_{q^m}$ over $\F_q$. The matrix \textbf{associated} to a vector
$\alpha \in \F_{q^m}^k$ with respect to $\mG$ is the $k \times m$ matrix 
$M_{\mG}(\alpha)$ with entries in $\F_q$ defined by $\alpha_i= \sum_{j=1}^m
M_{\mG}(\alpha)_{ij}\gamma_j$ for all $i=1,...,k$. The Delsarte code
\textbf{associated} to a Gabidulin code $C \subseteq \F_{q^m}^k$ with respect to
the basis $\mG$
 is $\mC_{\mG}(C):=\{ M_{\mG}(\alpha) : \alpha \in C\}
\subseteq
 \mbox{Mat}(k \times m, \F_q)$.
\end{definition}

Notice that, in the previous definition, the $i$-th row of $M_{\mG}(\alpha)$ is just the expansion
of the entry $\alpha_i$ over the basis $\mG$. The following result is immediate
and well-known. We include it here
for completeness.

\begin{proposition}\label{gabtodel}
 Let $C \subseteq \F_{q^m}^k$ be a Gabidulin code. For any basis $\mG$ of $\F_{q^m}$ over $\F_q$, $\mC_{\mG}(C) \subseteq
\mbox{Mat}(k\times m, \F_q)$
is a Delsarte rank-metric code with
 $$\dim_{\F_q}
\mC_{\mG}(C)=m \cdot \dim_{\F_{q^m}}(C).$$
 Moreover, $\mC_{\mG}(C)$ has the 
 same rank distribution as $C$. In particular we have $\mbox{maxrk}(C)=\mbox{maxrk}(\mC_{\mG}(C))$, and if $C \neq 0$   
 then $\mbox{minrk}(C)= \mbox{minrk}(\mC_{\mG}(C))$.

\end{proposition}

\begin{remark}
 Proposition \ref{gabtodel} shows that any Gabidulin code can be regarded as a Delsarte rank-metric code with the same cardinality
 and rank distribution. Clearly, since Gabidulin codes are $\F_{q^m}$-linear spaces and Delsarte codes
are $\F_q$-linear spaces, not all
 Delsarte rank-metric codes arise from a Gabidulin code in this way. In fact, only a few of them do. For example, a Delsarte code
 $\mC \subseteq \mbox{Mat}(k \times m, \F_q)$ such that $\dim_{\F_q}(\mC)
\not\equiv 0 \mod m$ cannot arise from a Gabidulin code.
\end{remark}

In the remainder of the section we compare the duality theories of
Delsarte and Gabidulin codes, proving in particular that the former
generalizes the latter.

\begin{remark} \label{prob}
 Given a Gabidulin code $C \subseteq \F_{q^m}^k$ and a basis $\mG$ of $\F_{q^m}$ over $\F_q$,
it is natural to ask whether the Delsarte codes $\mC_{\mG}(C^\perp)$ and 
$\mC_{\mG}(C)^\perp$ coincide or not. The answer is unfortunately negative in general, as we show 
in the following example.
\end{remark}

\begin{example}
 Let $q=3$, $k=m=2$ and
$\F_{3^2}=\F_3[\eta]$, where
$\eta$ is a root of the irreducible primitive polynomial $x^2+2x+2 \in \F_3[x]$.
Let $\xi:=\eta^2$, so that $\xi^2+1=0$.
Set $\alpha:=(\xi,2)$, and let $C \subseteq \F_{3^2}^2$ be the 1-dimensional
Gabidulin code generated by $\alpha$ over $\F_{3^2}$.  Take $\mG:= \{ 1,\xi\}$
as 
basis of $\F_{3^2}$ over $\F_3$. One can easily check that $\mC_{\mG}(C)
\subseteq \mbox{Mat}(2\times 2,\F_3)$
is generated over $\F_3$ by the two matrices 
$$M_{\mG}(\alpha)= \begin{bmatrix} 
                  0 & 1 \\ 2 & 0 
                 \end{bmatrix}, \ \ \ \ \ \ \  M_{\mG}(\xi \alpha)=
\begin{bmatrix} 
                  -1 & 0 \\ 0 & 2 
                 \end{bmatrix}.
                 $$
Let $\beta:=(\xi,1) \in \F_{3^2}^2$. We have $\langle \alpha, \beta
\rangle = 1 \neq 0$, and so $\beta \notin C^\perp$. It follows 
$M_{\mG}(\beta) \notin \mC_{\mG}(C^\perp)$.
On the other
hand, 
$$ M_\mG(\beta)=
\begin{bmatrix} 
                  0 & 1 \\ 1 & 0 
                 \end{bmatrix},$$
and it is easy to see that $M_\mG(\beta)$ is trace-orthogonal to both
$M_\mG(\alpha)$ and $M_\mG(\xi \alpha)$. It follows $M_\mG(\beta) \in
\mC_{\mG}(C)^\perp$, and so $\mC_{\mG}(C)^\perp \neq
\mC_{\mG}(C^\perp)$.
\end{example}

Although, for a fixed basis $\mG$, the duality notions for Delsarte and
Gabidulin codes do not coincide in general, we show that there is a simple
relation between them via orthogonal bases of finite fields.

\begin{definition}
 Let $\mbox{Trace}:\F_{q^m} \to \F_q$ be the $\F_q$-linear \textbf{trace} map given by
$\mbox{Trace}(\alpha):= \alpha +\alpha^q+ \cdots + \alpha^{q^{m-1}}$ for all
 $\alpha \in \F_{q^m}$.  Bases $\mG= \{ \gamma_1,...,\gamma_m\}$ and 
$\mG'= \{ \gamma'_1,...,\gamma'_m\}$ 
of $\F_{q^m}$ over
$\F_q$ are said to be \textbf{orthogonal} (or \textbf{dual}) if
 $\mbox{Trace}(\gamma'_i\gamma_j)=\delta_{ij}$ for all $i,j \in \{ 1,...,m\}$.
\end{definition}

The following result on orthogonal bases is well-known.

\begin{proposition}[\cite{niede}, page 54]
 For every basis $\mG$ of $\F_{q^m}$ over
$\F_q$ there exists a unique orthogonal basis $\mG'$.
\end{proposition}

\begin{theorem} \label{self}
Let $C \subseteq \F_{q^m}^k$ be a Gabidulin code, and let $\mG$, $\mG'$
be orthogonal bases of $\F_{q^m}$ over $\F_q$. We have
$$\mC_{\mG'}(C^\perp)= \mC_{\mG}(C)^\perp.$$
In particular, if we set $\mC:=\mC_{\mG}(C)$, then $C$ has the same rank
distribution as $\mC$, and $C^\perp$
has the same rank distribution as $\mC^\perp$.
\end{theorem}

\begin{proof} Let $\mG=\{ \gamma_1,...,\gamma_m\}$ and $\mG'=\{ \gamma'_1,...,\gamma'_m\}$. 
 Take any $M \in \mC_{\mG'}(C^\perp)$ and $N \in
\mC_{\mG}(C)$. There exist $\alpha \in C^\perp$ and $\beta \in C$ such that
$M=M_{\mG'}(\alpha)$ and $N=M_{\mG}(\beta)$. According to Definition
\ref{assoc} we have
\begin{equation} 
 0= \langle \alpha,\beta \rangle = \sum_{i=1}^k
\alpha_i\beta_i=\sum_{i=1}^k \sum_{j=1}^m M_{ij}\gamma'_j \sum_{t=1}^m
N_{it} \gamma_t=\sum_{i=1}^k \sum_{j=1}^m \sum_{t=1}^m
M_{ij}N_{it} \gamma'_j \gamma_t.\label{eqqq}
\end{equation}
Applying the function $\mbox{Trace}:\F_{q^m} \to \F_q$ to both sides of
equation (\ref{eqqq}) we get
\begin{eqnarray*} \label{ccc}
 0 &=& \mbox{Trace} \left( \sum_{i=1}^k \sum_{j=1}^m \sum_{t=1}^m
M_{ij}N_{it} \gamma'_j \gamma_t \right) \\
 &=&  \sum_{i=1}^k \sum_{j=1}^m \sum_{t=1}^m
M_{ij}N_{it} \mbox{Trace}(\gamma'_j \gamma_t)  \\
&=&  \sum_{i=1}^k \sum_{j=1}^m \sum_{t=1}^m
M_{ij}N_{it} \delta_{jt}  \\
&=&  \sum_{i=1}^k \sum_{j=1}^m 
M_{ij}N_{ij} \\
&=& \mbox{Tr}(MN^t) \\
&=& \langle M, N \rangle.
\end{eqnarray*}
It follows $\mC_{\mG'}(C^\perp) \subseteq
\mC_{\mG}(C)^\perp$. By Proposition \ref{gabtodel} and Lemma \ref{proprduale},
$\mC_{\mG'}(C^\perp)$ and $\mC_{\mG}(C)^\perp$ have the same dimension over
$\F_q$.
Hence the two codes are equal. The second part of the statement easily follows
from Proposition \ref{gabtodel}.
\end{proof}

\begin{remark} \label{delsbattegab}
Theorem \ref{self} shows that the duality theory of Delsarte rank-metric codes can be 
regarded as a
generalization of the duality theory of Gabidulin rank-metric codes. In particular, we notice
that 
all the results on Delsarte codes which we will prove in the following sections also apply to
Gabidulin codes.
Note that a relation between the trace-product of $\mbox{Mat}(k \times m,
\F_q)$ and the standard inner product
of $\F_{q^m}^k$ involving orthogonal bases was also pointed out in
\cite{grant0}.
\end{remark}

In the remainder of the paper we focus on the general case of Delsarte rank-metric codes.

\section{MacWilliams identities for rank-metric codes} \label{sec3}
In this section we give an elementary proof of certain MacWilliams identities
for
Delsarte rank-metric codes.
MacWilliams identities for  such codes  were also obtained in \cite{del1} by
Delsarte himself using the machinery of association schemes. The formulas which
we derive are different from those of \cite{del1}, but are significantly more
straightforward. Indeed, the proof that we present is elementary and
concise, and only
employs linear algebra and a double
counting argument. In Appendix \ref{appa} we also show how the original
formulas by Delsarte can be obtained from our formulas as a corollary. A
different formulation of the identities of \cite{del1} can be found in 
\cite{GYj}.

\begin{definition} \label{gaussian}
 Let $q$ be a prime power, and let $s$ and $t$ be integers. The
$q$-\textbf{binomial coefficient} of $s$ and $t$ is denoted and defined
by
 $$\displaystyle{\begin{bmatrix} s \\ t\end{bmatrix}_q = \left\{
\begin{array}{ll}
                                         0 & \mbox{ if $s<0$, $t<0$, or $t>s$,}
\\
                                         1 & \mbox{ if $t=0$ and $s \ge 0$}, \\
                                         \prod_{i=1}^t \frac{q^{s-i+1}-1}{q^i-1}& \mbox{ otherwise.}
                                        \end{array} \right.\ }
$$
\end{definition}

 It is well-known that this number counts
the number
of $t$-dimensional $\F_q$-subspaces of an $s$-dimensional $\F_q$-space. In
particular we have $$\begin{bmatrix} s \\ t\end{bmatrix}_q = \begin{bmatrix} s
\\ s-t\end{bmatrix}_q$$ for all integers $s,t$. Since
in this paper we work with a fixed prime power $q$, we omit the subscript 
in the sequel.

\begin{remark}\label{spazio}
 Given any matrices $M,N \in \mbox{Mat}(k \times m, \F_q)$ we always have
$\mbox{colsp}(M+N) \subseteq \mbox{colsp}(M)+\mbox{colsp}(N)$. As a
consequence, 
if $U \subseteq \F_q^k$ is a vector subspace, then the set of matrices $M \in
\mbox{Mat}(k \times m, \F_q)$ 
with $\mbox{colsp}(M) \subseteq U$ is a vector subspace of $\mbox{Mat}(k
\times m, \F_q)$.
\end{remark}

\begin{notation}
 We denote the vector space $\{M \in
\mbox{Mat}(k \times m, \F_q) : \mbox{colsp}(M) \subseteq U\}$ of Remark
\ref{spazio} by $\mbox{Mat}_U(k \times m, \F_q)$.
\end{notation}

We start with a series of preliminary results.

\begin{lemma} \label{dimensioni}
 Let  $U \subseteq \F_q^k$ be a subspace.  We have $\dim_{\F_q} \mbox{Mat}_U(k
\times m, \F_q)=m \cdot \dim_{\F_q}(U)$.
\end{lemma}
\begin{proof}
 Let $s:=\dim_{\F_q}(U)$. Define the $s$-dimensional  space $V:=\{
x \in \F_q^k : x_i=0 \mbox{ for } i >s \} \subseteq \F_q^k$. There exists an
$\F_q$-isomorphism $g:\F_q^k \to \F_q^k$ that maps $U$ into $V$. Let $G \in
\mbox{Mat}(k \times k, \F_q)$ be the invertible matrix associated to $g$ with
respect to the canonical
basis $\{e_1,...,e_k\}$ of $\F_q^k$, i.e.,
$$g(e_j) = \sum_{i=1}^k G_{ij} e_i \ \ \ \ \ \mbox{ for all }
j=1,...,k.$$ 
For any matrix $M \in \mbox{Mat}(k\times m, \F_q)$ we
have
$g(\mbox{colsp}(M))=\mbox{colsp}(GM)$, and it is easy to check that the map
$M \mapsto GM$ is an $\F_q$-isomorphism
$\mbox{Mat}_U(k\times m, \F_q) \to \mbox{Mat}_V(k\times m, \F_q)$. Now we observe that 
$\mbox{Mat}_V(k\times m, \F_q)$
is the vector space of matrices $M \in \mbox{Mat}(k\times m, \F_q)$ whose
last $k-s$ rows equal zero.
Hence $\dim_{\F_q} \mbox{Mat}_V(k\times m, \F_q)=km-m(k-s)=ms$, and the lemma
follows.
\end{proof}

\begin{lemma} \label{iduali}
 Let  $U \subseteq \F_q^k$ be a subspace. 
 We have $\mbox{Mat}_U(k \times m, \F_q)^\perp=\mbox{Mat}_{U^\perp}(k \times m, \F_q)$.
\end{lemma}

\begin{proof}
 Let $N \in \mbox{Mat}_{U^\perp}(k \times m, \F_q)$ and $M \in \mbox{Mat}_U(k \times m, \F_q)$. By definition, each
 column of $N$ belongs to $U^\perp$, and
each column of $M$ belongs to
$U$. Hence by Lemma \ref{rel} we have
 $$\langle M, N \rangle =\sum_{i=1}^m  \langle M_i, N_i \rangle =0.$$ 
This proves $\mbox{Mat}_{U^\perp}(k \times m, \F_q) \subseteq \mbox{Mat}_U(k
\times m, \F_q)^\perp$. By Lemma \ref{dimensioni}, the two spaces of matrices 
$\mbox{Mat}_{U^\perp}(k \times m, \F_q)$ and $\mbox{Mat}_U(k
\times m, \F_q)^\perp$ have the same dimension over $\F_q$. Hence they are
equal.
\end{proof}

\begin{lemma} \label{tecn}
 Let $\mC \subseteq \mbox{Mat}(k \times m, \F_q)$ be a code, and let $U
\subseteq \F_q^k$ be a subspace. Denote by $s$ the
 dimension of $U$ over $\F_q$.
  We have
 $$|\mC \cap \mbox{Mat}_U(k \times m, \F_q)|= \frac{|\mC|}{ q^{m(k-s)}}
|\mC^\perp \cap \mbox{Mat}_{U^\perp}(k \times m, \F_q)|.$$
\end{lemma}
\begin{proof}
 Combining Lemma \ref{proprduale} and Lemma \ref{iduali} we obtain
 $$(\mC \cap \mbox{Mat}_U(k \times m, \F_q))^\perp=\mathcal{\mC}^\perp
+\mbox{Mat}_U(k\times m, \F_q)^\perp=\mC^\perp + \mbox{Mat}_{U^\perp}(k \times
m, \F_q).$$ Hence by Lemma \ref{proprduale} we have
 \begin{equation} \label{eq2}
 |\mC \cap \mbox{Mat}_U(k \times m, \F_q)| \cdot |\mC^\perp + \mbox{Mat}_{U^\perp}(k \times m, \F_q)|=q^{km}.
 \end{equation}
 On the other hand, Lemma \ref{dimensioni} gives
 $$\dim_{\F_q}(\mC^\perp + \mbox{Mat}_{U^\perp}(k \times m,
\F_q))=\dim_{\F_q}(\mC^\perp)+m \cdot \dim_{\F_q}U^\perp-\dim_{\F_q} 
 (\mC^\perp \cap \mbox{Mat}_{U^\perp}(k \times m, \F_q)),$$ and so, again by
Lemma \ref{proprduale},
  \begin{equation}\label{eq3}
|\mC^\perp + \mbox{Mat}_{U^\perp}(k \times m, \F_q)| = \frac{q^{km} \cdot
q^{m(k-s)}}{|\mC| \cdot
|\mC^\perp \cap
\mbox{Mat}_{U^\perp}(k \times m, \F_q)|}.
 \end{equation}
 Combining equation (\ref{eq2}) and equation (\ref{eq3}) one easily obtains the
lemma.
\end{proof}

The following result is well-known, but we include it for completeness.

\begin{lemma} \label{facile}
 Let $0 \le t,s \le k$ be integers, and let $X \subseteq \F_q^k$ be a subspace
of dimension $t$ over $\F_q$. The number of subspaces $U \subseteq \F_q^k$ such
that $X
\subseteq U$ and $\dim_{\F_q}(U)=s$ is
$$\begin{bmatrix} k-t \\ s-t
\end{bmatrix}.$$
\end{lemma}
\begin{proof}
 Let $\pi: \F_q^k \to \F_q^k/X$ denote the projection on the quotient vector
space $\F_q^k$ modulo $X$. It is easy to see that $\pi$ induces a bijection
between the $s$-dimensional vector subspaces of $\F_q^k$ containing $X$
and the $(s-t)$-dimensional subspaces of $\F_q^k/X$. The lemma follows from the
fact that $\F_q^k/X$ has dimension $k-t$.
\end{proof}

\begin{lemma} \label{subitoprima}
 Let $\mC \subseteq \mbox{Mat}(k \times m, \F_q)$ be a code. Denote by
 $(A_i)_{i \in \N}$ the rank distribution of $\mC$. Let $0 \le s \le k$ be an integer.
 We have
 $$\sum_{\substack{U \subseteq \F_q^k \\ \dim_{\F_q}(U)=s}} |\mC \cap \mbox{Mat}_U(k \times m, \F_q)| = 
 \sum_{i=0}^k A_i \begin{bmatrix} k-i \\ k-s \end{bmatrix}.$$
\end{lemma}

\begin{proof}
 Define the set $\mathcal{A}(\mC,s) := \{ (U,M) : U \subseteq \F_q^k, \ \dim(U)=s, \ M
\in \mC, \ \mbox{colsp}(M) \subseteq U \}$.
We will count the elements of $\mA(\mC,s)$ in two different ways.
On the one
hand, using  Lemma \ref{facile}, we have

\begin{eqnarray*}
 |\mA(\mC,s)| &=& \sum_{M \in \mC} |\{ U \subseteq \F_q^k, \ \dim(U)=s, \
\mbox{colsp}(M)
\subseteq U \}| \\
&=&  \sum_{i=0}^k \sum_{\substack{M \in \mC \\ \mbox{\small{rk}}(M)=i}} |\{ U
\subseteq \F_q^k, \ \dim(U)=s, \ \mbox{colsp}(M)
\subseteq U \}| \\
&=&  \sum_{i=0}^k \sum_{\substack{M \in \mC \\ \mbox{\small{rk}}(M)=i}}
\begin{bmatrix} k-i \\ s-i \end{bmatrix}  \; = \; \sum_{i=0}^{k} A_i
\begin{bmatrix} k-i \\ s-i \end{bmatrix} = \; \sum_{i=0}^{k} A_i
\begin{bmatrix} k-i \\ k-s \end{bmatrix}.
\end{eqnarray*}
On the other hand, 
$$
|\mA(\mC,s)| = \sum_{\substack{U \subseteq \F_q^k \\ \dim(U) = s}} |\{ M \in
\mC : \mbox{colsp}(M) \subseteq U \}| = \sum_{\substack{U \subseteq \F_q^k \\
\dim(U)=s}} |\mC \cap \mbox{Mat}_U(k \times m, \F_q)|,
$$
and the lemma follows.
\end{proof}

Now we state our main result.

\begin{theorem} \label{theo}
 Let $\mathcal{C} \subseteq \mbox{Mat}(k \times m, \F_q)$ be a code. Let
${(A_i)}_{i \in \N}$ and ${(B_j)}_{j \in \N}$ be the rank distributions of $\mC$ and $\mC^\perp$,
respectively. For any integer $0 \le \nu \le k$ we have
$$\sum_{i=0}^{k-\nu} A_i \begin{bmatrix} k-i \\ \nu \end{bmatrix} = \;
\frac{|\mC|}{q^{m\nu}} \; \sum_{j=0}^{\nu} B_j \begin{bmatrix} k-j \\ \nu-j
\end{bmatrix}.$$
\end{theorem}

\begin{proof}
  Lemma \ref{subitoprima} applied to $\mC$ with $s=k-\nu$ gives
 $$\sum_{\substack{U \subseteq \F_q^k \\ \dim_{\F_q}(U)=k-\nu}} |\mC \cap \mbox{Mat}_U(k \times m, \F_q)| \; = \; 
 \sum_{i=0}^k A_i \begin{bmatrix} k-i \\ \nu \end{bmatrix}.$$
 The map $U \mapsto U^\perp$ is a bijection between the $\nu$-dimensional and
  the $(k-\nu)$-dimensional subspaces of $\F_q^k$. Hence we have
 $$\sum_{\substack{U \subseteq \F_q^k \\ \dim_{\F_q}(U)=k-\nu}} |\mC^\perp \cap \mbox{Mat}_{U^\perp}(k \times m, \F_q)| \; =  
   \sum_{\substack{U \subseteq \F_q^k \\ \dim_{\F_q}(U)=\nu}} |\mC^\perp \cap \mbox{Mat}_{U}(k \times m, \F_q)| 
  \; = \; \sum_{j=0}^k B_j \begin{bmatrix} k-j \\ k-\nu \end{bmatrix},$$
where the second equality follows from Lemma \ref{subitoprima} applied
 to the code $\mC^\perp$ with $s=\nu$.
 Lemma \ref{tecn} with $s=k-\nu$ gives
$$\sum_{i=0}^{k} A_i \begin{bmatrix} k-i \\ \nu \end{bmatrix} = \;
\frac{|\mC|}{q^{m\nu}} \; \sum_{j=0}^{k} B_j \begin{bmatrix} k-j \\ \nu-j
\end{bmatrix}.$$
By definition, for $i > k-\nu$ and for $j > \nu$ we have
$$\begin{bmatrix} k-i \\ \nu \end{bmatrix}=\begin{bmatrix} k-j \\ \nu-j
\end{bmatrix}=0,$$ and the theorem follows.
\end{proof}

\begin{remark}
 Theorem \ref{theo} can be regarded as the $q$-analog of Lemma 2.2 of \cite{macsys},
 which yields analogous  identities for
the Hamming metric.
\end{remark}

Theorem \ref{theo} produces in particular MacWilliams-type identities that
relate the rank
distribution of a dual code $\mC^\perp$ to the rank distribution
of $\mC$. The following result gives a recursive method to compute the
rank distribution of $\mC^\perp$ from the rank distribution of $\mC$.

\begin{corollary} \label{formule}
 Let $\mathcal{C} \subseteq \mbox{Mat}(k \times m, \F_q)$ be a code.
 Let
${(A_i)}_{i \in \N}$ and ${(B_j)}_{j \in \N}$ be the rank distributions of $\mC$ and $\mC^\perp$,
respectively. For $\nu=0,...,k$ define
$$ a_{\nu}^k:= \frac{q^{m\nu}}{|\mC|} \; \sum_{i=0}^{k-\nu} A_i
\begin{bmatrix} k-i \\ \nu \end{bmatrix}.$$
The $B_j$'s are given by the recursive formula

$$
\left\{ \begin{array}{l}
           B_0=1, \\ B_{\nu}= a_\nu^k - \mathlarger{\sum}_{j=0}^{\nu-1} B_j
\begin{bmatrix}
                                                             k-j \\ \nu-j
                                                            \end{bmatrix} \ \mbox{ for } \nu=1,...,k, \\
                                                            B_\nu=0 \ \mbox{ for } \nu > k.

          \end{array}
\right.\ 
$$
\end{corollary}

\begin{proof}
 Clearly, $B_0=1$ and $B_\nu=0$ for $\nu >k$. For any fixed integer $\nu \in \{1,...,k\}$ Theorem \ref{theo} gives
 $$a_\nu^k = \sum_{j=0}^{\nu-1} B_j \begin{bmatrix} k-j \\ \nu-j\end{bmatrix} +
B_\nu,$$
 which proves the result.
\end{proof}

\begin{remark}
 Identities in the form of Theorem \ref{theo} are usually called ``moments of
MacWilliams
identities'' rather than ``MacWilliams identities''. For convenience, in this
paper we will call 
``MacWilliams identities'' both the identities of Theorem \ref{theo} and Corollary 
\ref{formule}. 
\end{remark}

\begin{remark}
 We notice that Theorem \ref{theo} implies 
Theorem 3.3 of \cite{del1} as a corollary (see Appendix \ref{appa} for details),
producing
MacWilliams 
identities for Delsarte codes in an explicit form employing an elementary
argument.
\end{remark}

\begin{remark}
 Identities in the form of Theorem \ref{theo}
were recently proved for Gabidulin
  codes (see \cite{gadu}, Proposition 3). The proof of \cite{gadu} is based on
the Hadamard transform, $q$-products, $q$-derivatives 
 and $q$-transforms of polynomials. In \cite{GYj}, Corollary 1 and Proposition
3, the authors
show that such identities also apply to Delsarte codes. Their proof is based on
the results of \cite{del1} by Delsarte.
\end{remark}

Theorem \ref{theo} and Corollary \ref{formule} allow us to 
re-establish the main results
of the duality theory of rank-metric codes in a very concise way.

\begin{corollary} \label{coro1}
 The rank distribution of a code $\mC$  determines the rank distribution
of the dual code $\mC^\perp$.
\end{corollary}

\begin{proof}
 This immediately follows from Corollary \ref{formule}.
\end{proof}

\begin{remark}
 Corollary \ref{coro1} was first proved by Delsarte using the theory of
association schemes. See \cite{del1}, Theorem 3.3 for details.
\end{remark}

\begin{example}
 Let $q=3$, $k=3$, $m=4$. Consider the code $\mC \subseteq \mbox{Mat}(3\times
4,\F_3)$ generated by the following three matrices:
$$ \begin{bmatrix}  
    1 & 2 & 0 & 0 \\ 0 & 1 & 0 & 0 \\ 0 & 0 & 2 & 1
   \end{bmatrix}, \ \ \ \ \begin{bmatrix}  
    0 & 2 & 0 & 0 \\ 0 & 0 & 1 & 2 \\ 1 & 1 & 0 & 0
   \end{bmatrix}, \ \ \ \ \begin{bmatrix}  
    0 & 2 & 0 & 0 \\ 0 & 0 & 1 & 2 \\ 1 & 1 & 1 & 1
   \end{bmatrix}.$$
It can be checked that $\dim_{\F_3} \mC=3$ and that the rank distribution of
$\mC$ is
$A_0=1$, $A_1=2$, $A_2=0$, $A_3=24$. If ${(B_j)}_{j \in \N}$ denotes the rank
distribution of $\mC^\perp$, then the recursive formula of
Corollary \ref{formule} allows us to compute:
$$B_0=1, \ \ \ B_1=50, \ \ \ B_2=3432, \ \ \ B_3= 16200.$$
Notice that $\sum_{i=0}^3 B_i=19683=3^9=|\mC^\perp|$, as expected.
\end{example}

\begin{remark}
 For a code $\mC \subseteq \mbox{Mat}(k\times m, \F_q)$ define
 $\mC^t:= \{ M^t : M \in \mC\} \subseteq \mbox{Mat}(m\times k, \F_q)$. Clearly,
$\mC$ and $\mC^t$ have the same dimension and rank distribution. Moreover, one can
check that $(\mC^t)^\perp = (\mC^\perp)^t$. As a consequence, up to a
transposition, without loss of generality in the sequel we will always assume $k\le
m$ in the proofs of our results.
\end{remark}

\begin{corollary} \label{dualmrd}
 If a code $\mC$ is MRD, then $\mC^\perp$ is also MRD.
\end{corollary}

\begin{proof}
 Let $\mC \subseteq \mbox{Mat}(k\times m, \F_q)$ be MRD. If $\mC =\{0\}$ or 
$\mC=\mbox{Mat}(k\times m, \F_q)$ the result follows from
Definition \ref{defmrd} and Remark \ref{tuttomrd}.
 Hence we assume $0 < \dim_{\F_q}(\mC)<km$. Assume $k \le m$
without loss of generality. 
Denote by $d$ the minimum rank of $\mC$, so 
 that $|\mC|=q^{m(k-d+1)}$.
 Let ${(A_i)}_{i \in \N}$ and ${(B_j)}_{j \in \N}$ be the rank distributions of
 $\mC$ and $\mC^\perp$, respectively. We have $A_0=B_0=1$ and $A_i=0$ for $1 \le
i \le d-1$.  
 Theorem \ref{theo} with $\nu=k-d+1$ gives
 $$\begin{bmatrix} k \\ k-d+1   \end{bmatrix} =  \begin{bmatrix} k \\ k-d+1  
\end{bmatrix} +
 \sum_{j=1}^{k-d+1} B_j \begin{bmatrix} k-j \\ k-d+1-j   \end{bmatrix},
$$
i.e.,
$$\sum_{j=1}^{k-d+1} B_j \begin{bmatrix} k-j \\ k-d+1-j   \end{bmatrix}=0.$$
Since $d \ge 1$, for $1 \le j \le k-d+1$ we have $k-j \ge k-d+1-j \ge 0$, and so 
$\begin{bmatrix} k-j \\
k-d+1-j   \end{bmatrix}>0$. Hence it must be
$B_j=0$ for $1 \le j \le k-d+1$, i.e., $\mbox{minrk}(\mC^\perp) \ge k-d+2$.
On the other hand, Theorem \ref{bound1} gives $\dim_{\F_q}(\mC^\perp) =m(d-1)\le
m(k-\mbox{minrk}(\mC^\perp)+1)$, i.e., 
$\mbox{minrk}(\mC^\perp) \le k-d+2$. It follows $\mbox{minrk}(\mC^\perp)=k-d+2$,
and so $\mC^\perp$ is MRD.
\end{proof}

\begin{remark}
 Corollary \ref{dualmrd} was first proved by Delsarte using the
theory of
designs and codesigns in regular semilattices (\cite{del1}, 
Theorem 5.5).
Theorem \ref{theo} allows us to give a short proof for the same
result.
Notice also that, by Remark \ref{delsbattegab}, Corollary
\ref{dualmrd} generalizes the analogous
result for Gabidulin codes of \cite{gabid}.
\end{remark}

\section{Minimum and maximum rank of a code} \label{sec4}
In this section we investigate the minimum and the maximum rank of a Delsarte code $\mC$,
and show how they relate to the minimum and maximum rank of its dual code $\mC^\perp$.
As an application, we give a recursive formula for the rank distribution of an MRD code.

\begin{proposition} \label{+2}
 Let $\mC \subsetneq \mbox{Mat}(k\times m, \F_q)$ be a non-zero code. We have
 $$\mbox{minrk}(\mC^\perp) \le \min \{ k,m\} - \mbox{minrk}(\mC)+2.$$
 Moreover, the bound is attained if and only if $\mC$ is MRD.
\end{proposition}
\begin{proof}
 Assume $k \le m$ without loss of generality. Theorem
\ref{bound1} applied to the code $\mC$ gives $\dim_{\F_q}(\mC) \le
m(k-\mbox{minrk}(\mC)+1)$.
The same theorem applied to $\mC^\perp$ gives $\dim_{\F_q}(\mC^\perp) \le m(k-\mbox{minrk}(\mC^\perp)+1)$, i.e., 
 $\dim_{\F_q}(\mC) \ge m(\mbox{minrk}(\mC^\perp)-1)$. Hence we have
 \begin{equation} \label{eq1}
  m(\mbox{minrk}(\mC^\perp)-1) \le \dim_{\F_q}(\mC) \le m(k-\mbox{minrk}(\mC)+1).
 \end{equation}
 In particular, $\mbox{minrk}(\mC^\perp)-1 \le k-\mbox{minrk}(\mC)+1$, and the bound follows. Let us prove the second part of the
 statement. Assume that $\mC$ is MRD, and let $d:=\mbox{minrk}(\mC)$. We
have $\dim_{\F_q}(\mC)=m(k-d+1)$,
 and so $\dim_{\F_q}(\mC^\perp)=m(d-1)$. By Corollary \ref{dualmrd},
$\mC^\perp$ is also MRD, and so
$m(d-1)=m(k-\mbox{minrk}(\mC^\perp)+1)$.
 It follows $\mbox{minrk}(\mC^\perp)=k-d+2$. On the other hand, if 
$\mbox{minrk}(\mC^\perp) = k - \mbox{minrk}(\mC)+2$ then 
 both the inequalities in (\ref{eq1}) are in fact equalities, and so $\mC$ is MRD.
\end{proof}

\begin{corollary} \label{nondipende}
 The rank distribution of a non-zero MRD  code $\mC \subseteq \mbox{Mat}(k\times
m, \F_q)$ only depends on $k$, $m$ and $\mbox{minrk}(\mC)$.
 \end{corollary}

\begin{proof}
 Assume $k \le m$ without loss of generality. Let
$d:=\mbox{minrk}(\mC)$, and
let ${(A_i)}_{i \in \N}$ denote the rank distribution of $\mC$. 
 By Proposition \ref{+2}, $\mC^\perp$ 
 has minimum rank $k-d+2$. Hence the equations of Theorem \ref{theo} for 
$0 \le \nu \le k-d$ reduce
 to
 $$\begin{bmatrix} k \\ \nu \end{bmatrix} + \sum_{i=d}^{k-\nu} A_i
\begin{bmatrix}  k-i \\ \nu \end{bmatrix}=\frac{|\mC|}{q^{m\nu}} \begin{bmatrix}
k \\ \nu \end{bmatrix}, \ \ \ \ 0 \le \nu \le k-d.$$
 These identities give a linear system of $k-d+1$ equations in the $k-d+1$
unknowns $A_{d},...,A_{k}$. It is easy to see that the matrix associated to 
 the system is triangular with all 1's on the diagonal. In particular, the
solution to the system is unique. Hence $A_{d},...,A_{k}$ are uniquely
 determined by $k$, $m$ and $d$. Since $A_0=1$ and $A_i=0$ for $0 <i<d$ and 
for $i>k$, the thesis follows.
\end{proof}

\begin{remark}
 Corollary \ref{nondipende} was first proved by Delsarte by computing
explicitly the rank distribution of an MRD code, and then observing that
the
obtained formulas only depend on the parameters $m$, $k$, $d$ (see
\cite{del1}, Theorem 5.6). Corollary \ref{nondipende} allows us to give a concise proof
for the same result.
The rank distribution of Delsarte MRD codes
was also computed in \cite{mcg} employing elementary techniques.
\end{remark}

\begin{remark} \label{formrd} Using the same argument as Corollary \ref{formule} it is easy to
derive a
recursive formula for the rank distribution ${(A_i)}_{i \in \N}$ of a non-zero MRD code $\mC
\subseteq
\mbox{Mat}(k\times m, \F_q)$ of given minimum rank $d$:

$$
 \left\{ \begin{array}{l}  A_0= 1, \ \ \ A_d = (q^m-1)
\begin{bmatrix} k \\ k-d \end{bmatrix} , \\
A_{d+\ell} = (q^{m(1+\ell)} -1)
\begin{bmatrix} k \\ k-d-\ell \end{bmatrix}  - 
\mathlarger{\sum}_{i=d}^{d+\ell-1} A_i \begin{bmatrix} k-i
\\ k-d-\ell \end{bmatrix}\ \ \ \
\mbox{ for $1 \le \ell \le k-d$.} \end{array}\right.\ 
$$
We do not go into the details of the proof. 
\end{remark}

The following result is the analogue of Theorem \ref{bound1} for the maximum
rank.

\begin{proposition} \label{anticode}
 Let $\mC \subseteq \mbox{Mat}(k\times m, \F_q)$ be a code. We have
 $$\dim_{\F_q}(\mC) \le \max \{ k,m\} \cdot \mbox{maxrk}(\mC).$$
 Moreover, for any choice of $0 \le D \le \min \{ k,m\}$ there exists a code 
 $\mC \subseteq \mbox{Mat}(k \times m, \F_q)$ with maximum rank
equal to $D$ and attaining the upper bound.
\end{proposition}

\begin{proof}
 Assume $k
\le m$ without loss of generality. Fix $0 \le D \le k$. The set of all $k \times
m$ matrices having the last
 $k-D$ rows equal to zero is an example of a code of maximum rank   $D$ and dimension $mD$ over $\F_q$. Now we prove the
 first part of the statement.
 Let $\mC \subseteq \mbox{Mat}(k \times m, \F_q)$ be a code with $\mbox{maxrk}(\mC)=D$. If $D=k$ then the bound is trivial. Hence we assume 
 $D \le k-1$. Theorem \ref{bound1} gives a code $\mD \subseteq
\mbox{Mat}(k\times m, \F_q)$ with $\mbox{minrk}(\mD) = D+1$ and 
 $\dim_{\F_q}(\mD)=m(k-D)$. We clearly have $\mC \cap \mD =\{0\}$ and $\mC \oplus
\mD \subseteq \mbox{Mat}(k\times m, \F_q)$. Hence $\dim_{\F_q}(\mC) \le
km-\dim_{\F_q}(\mD)=mD$.
\end{proof}

\begin{definition} \label{defanticode}
 A code $\mC \subseteq \mbox{Mat}(k\times m, \F_q)$ which attains the upper bound of 
Proposition \ref{anticode}
 is said to be a (\textbf{Delsarte}) \textbf{optimal anticode}.
\end{definition}

We conclude the section with a result that relates the minimum rank of a code
with the maximum rank of the dual code.

\begin{proposition}
 Let $\mC \subseteq \mbox{Mat}(k\times m, \F_q)$ be a non-zero code. We have
 $$\mbox{minrk}(\mC) \le \mbox{maxrk}(\mC^\perp)+1.$$
\end{proposition}

\begin{proof}
 Assume $k \le m$ without loss of generality.
Applying Theorem
\ref{bound1} to $\mC$ we obtain 
 $\dim_{\F_q}(\mC) \le m(k-\mbox{minrk}(\mC)+1)$, while Proposition \ref{anticode} applied to $\mC^\perp$ gives
 $\dim_{\F_q}(\mC^\perp) \le m \cdot \mbox{maxrk}(\mC^\perp)$, i.e., 
$\dim_{\F_q}(\mC) \ge m(k-\mbox{maxrk}(\mC^\perp))$.
 Hence we have
 $$m(k-\mbox{maxrk}(\mC^\perp)) \le \dim_{\F_q}(\mC) \le m(k-\mbox{minrk}(\mC)+1),$$
 and the thesis follows.
\end{proof}

\section{Optimal anticodes} \label{sec4a}

In this section we provide a new characterization of  optimal
anticodes in terms of their
intersection with MRD codes. As an application of such a description, we prove
that the dual of an
optimal anticode is an optimal anticode.

Let us first briefly recall some notions which we will need in
the sequel. See \cite{niede}, Section 3.4 for details.

\begin{definition}
 Let $\F_{q^m}/\F_q$ be a finite field extension. A
\textbf{linearized polynomial}
$p$ over $\F_{q^m}$ is a polynomial of the form
$$p(x)=\alpha_0x + \alpha_1 x^q+ \alpha_2 x^{q^2}+ \cdots +\alpha_{s}x^{q^s},
\ \ \ \ \ \alpha_i \in \F_{q^m}, \ \ i=0,...,s.$$
The \textbf{degree} of $p$, denoted by $\deg(p)$,  is the largest $i \ge 0$
such that $\alpha_i \neq 0$.
\end{definition}

\begin{remark}
 It is
well known (\cite{niede}, Theorem 3.50) that the roots
of a linearized polynomial $p$ over $\F_{q^m}$ form an $\F_q$-vector subspace of
$\F_{q^m}$,
which we denote by
$V(p) \subseteq \F_{q^m}$. Notice that for any linearized polynomial $p$ we have
$\dim_{\F_q}V(p) \le \deg(p)$.
\end{remark}

\begin{lemma} \label{prel}
 Let $\mC \subseteq \mbox{Mat}(k\times m, \F_q)$ be a non-zero MRD code with
minimum rank $d$, and let ${(A_i)}_{i \in \N}$ be the rank distribution of $\mC$. Then
$A_{d+\ell}>0$ for all $0 \le \ell \le \min \{k,m\}-d$.
\end{lemma}

\begin{proof}
 Assume $k\le m$ without loss of generality. By
Corollary \ref{nondipende}, we shall prove the lemma for a given  MRD code $\mC
\subseteq \mbox{Mat}(k\times m, \F_q)$ of our choice with minimum rank
$d$. We first construct a convenient MRD code with the prescribed parameters, and
we essentially follow the construction of \cite{gabid}.

Let
$\gamma_1,...,\gamma_k \in \F_{q^m}$ be linearly independent over $\F_q$. Denote
by $\mathcal{L}(\F_{q^m},k-d)$ the $\F_{q^m}$-vector space of linearized polynomials over
$\F_{q^m}$ of degree less than or equal to $k-d$. We have
$\dim_{\F_{q^m}} \mathcal{L}(\F_{q^m},k-d)=k-d+1$.
Let $ev:\mathcal{L}(\F_{q^m},k-d) \to \F_{q^m}^k$ be the evaluation map
defined by $ev(p):=(p(\gamma_1),...,p(\gamma_k))$ for any $p \in 
\mathcal{L}(\F_{q^m},k-d)$. Then the image of $ev$ is a Gabidulin code 
$C \subseteq \F_{q^m}^k$ with minimum rank $d$ and dimension $k-d+1$ over
$\F_{q^m}$ (\cite{KK1}, Theorem 14). Let $\mG$ be any basis of
$\F_{q^m}$ over $\F_q$. By Proposition \ref{gabtodel},
$\mC:=\mC_{\mG}(C) \subseteq \mbox{Mat}(k\times m,\F_q)$ is a Delsarte rank-metric
code with $\dim_{\F_q}(\mC)=m(k-d+1)$ and the same rank distribution as $C$. 
In particular, $\mC$ is a non-zero MRD code
with minimum rank $d$.

Now we prove the lemma for the MRD code $\mC$ that we constructed. Fix 
$0 \le \ell \le k-d$. Define $t:=k-d-\ell$, and let $U
\subseteq \F_{q^m}$ be the $\F_q$-subspace generated by $\{
\gamma_1,...,\gamma_{t}\}$. If $t=0$ we set $U$ to be the
zero space.
By \cite{niede}, Theorem 3.52, 
$$p_U:= \prod_{\beta \in U} (x-\beta) \in \F_{q^m}[x]$$
is a linearized polynomial over $\F_{q^m}$ of degree $t=k-d-\ell \le k-d$.
Hence $p_U \in  \mathcal{L}(\F_{q^m},k-d)$.  By Proposition \ref{gabtodel} it
suffices to prove that $ev(p_U)=(p_U(\gamma_1),...,p_U(\gamma_k))$ has rank
$d+\ell=k-t$. Clearly, $V(p_U)=U$. In particular we have
$ev(p_U)=(0,...,0,p_U(\gamma_{t+1}),...,p_U(\gamma_k))$. We will prove that 
$p_U(\gamma_{t+1}),..., p_U(\gamma_k)$ are linearly independent over $\F_q$.
Assume that there exist $a_{t+1},...,a_k \in \F_q$ such 
that $\sum_{i=t+1}^k a_i p_U(\gamma_i)=0$. Then we have $p_U \left(
\sum_{i=t+1}^k
a_i\gamma_i \right)=0$, i.e.,  $\sum_{i=t+1}^k
a_i\gamma_i \in V(p_U)=U$.
It follows that there exist $a_1,...,a_t \in \F_q$ such that
$\sum_{i=1}^t a_i \gamma_i = \sum_{i=t+1}^k a_i\gamma_i$, i.e., $\sum_{i=1}^t
a_i \gamma_i - \sum_{i=t+1}^k a_i\gamma_i=0$. Since $\gamma_1,...,\gamma_k$ are
linearly independent over $\F_q$, we have $a_i=0$ for all $i=1,...,k$. In
particular $a_i=0$ for $i=t+1,...,k$. Hence
$p_U(\gamma_{t+1}),...,p_U(\gamma_k)$ are linearly independent over $\F_q$, as claimed.
\end{proof}

In the following result we give a necessary and sufficient condition for a Delsarte code
$\mC \subseteq
\mbox{Mat}(k\times m,\F_q)$ with $\dim_{\F_q}(\mC) \equiv 0 \mod m$ to be an optimal anticode.

\begin{proposition} \label{criterio}
 Let $0 \le D \le \min \{ k,m\}-1$ be an integer, and let $\mC \subseteq
\mbox{Mat}(k\times m,\F_q)$ be an $\F_q$-subspace with $\dim_{\F_q}(\mC) =\max
\{k,m \}
\cdot D$. The following facts are equivalent.

\begin{enumerate}
 \item $\mC$ is an optimal anticode.
 \item $\mC \cap \mD =\{0\}$ for all non-zero MRD codes $\mD \subseteq
\mbox{Mat}(k\times m,\F_q)$ with $\mbox{minrk}(\mD)=D+1$.
\end{enumerate}
\end{proposition}

\begin{proof}
 If $\mC$ is an optimal anticode, then by Definition \ref{defanticode} we have
$D=\mbox{maxrk}(\mC)$. Hence if $\mD$ is any non-zero code with
$\mbox{minrk}(\mD)=D+1$ we clearly have $\mC \cap \mD = \{0\}$. So $(1)
\Rightarrow (2)$ is trivial. Let us prove $(2) \Rightarrow (1)$. By
contradiction, assume that $\mC$ is not an optimal anticode. Since
$\mbox{maxrk}(\mC) \ge D$ (see Proposition \ref{anticode}), we must have
$s:=\mbox{maxrk}(\mC) \ge D+1$. Let $N \in \mC$ with
$\mbox{rk}(N) =s$. Let $\mD'$ be a non-zero MRD code with
$\mbox{minrk}(\mD')=D+1$ (see Theorem \ref{bound1} for the existence of such a code). 
By Lemma \ref{prel} there exists $A \in \mD'$ with
$\mbox{rk}(A)=s$. There exist invertible matrices $P$ and $Q$ of size $k \times
k$ and $m\times m$ (respectively) such that $N=PAQ$. Define $\mD:=P\mD'Q:= \{
PMQ : M \in \mD' \}.$ Then $\mD \subseteq \mbox{Mat}(k\times m, \F_q)$ is a
non-zero MRD
code with $\mbox{minrk}(\mD)=D+1$ and such that $N \in \mC \cap \mD$. Since
$\mbox{rk}(N)=s \ge D+1 \ge 1$, $N$ cannot be the zero matrix. This contradicts
the hypothesis.
\end{proof}

The following result may be regarded as the analogue of Corollary
\ref{dualmrd} for anticodes in the rank metric.

\begin{theorem} \label{dualanticode}
 If $\mC$ is an optimal anticode, then $\mC^\perp$ is also an optimal anticode.
\end{theorem}

\begin{proof}
 Let $\mC \subseteq \mbox{Mat}(k\times m, \F_q)$ be an optimal anticode with
$D:=\mbox{maxrk}(\mC)$. Assume $k \le m$ without
loss of generality. If $D=k$ then the result is trivial. Hence from now on
we assume $0
\le D \le k-1$. By Definition \ref{defanticode} we have $\dim_{\F_q}(\mC) = mD$,
and so $\dim_{\F_q}(\mC^\perp)=m(k-D)$.  By Proposition \ref{criterio} it
suffices to prove that $\mC^\perp \cap \mD = \{0\}$ for all non-zero MRD codes
$\mD \subseteq \mbox{Mat}(k\times m, \F_q)$ with $\mbox{minrk}(\mD)=k-D+1$. If
$\mD$ is such an MRD code, then we have
$\dim_{\F_q}(\mD)=m(k-(k-D+1)+1)=mD<mk$. Hence, by Proposition \ref{+2},
$\mD^\perp$ is an MRD code with $\mbox{minrk}(\mD^\perp)=k-(k-D+1)+2=D+1$.
Proposition \ref{criterio} gives $\mC \cap \mD^\perp = \{0\}$. Since
$\dim_{\F_q}(\mC)+\dim_{\F_q}(\mD^\perp)=mD+m(k-(D+1)+1)=mk$, it follows $\mC
\oplus \mD^\perp = \mbox{Mat}(k\times m, \F_q)$. Hence by Lemma \ref{proprduale}
we have $\mC^\perp \cap \mD =\{0\}$, as claimed.
\end{proof}

The following result shows how the maximum rank of a code $\mC$ and the
maximum rank of the dual code $\mC^\perp$ relate to each other.

\begin{proposition}
 Let $\mC \subseteq \mbox{Mat}(k\times m, \F_q)$ be a code. We have
 $$\mbox{maxrk}(\mC) \ge \min \{ k,m\}-\mbox{maxrk}(\mC^\perp).$$
 Moreover, the bound is attained if and only if $\mC$ is an
optimal anticode.
\end{proposition}
\begin{proof}
 Assume $k \le m$ without loss of generality. 
 Proposition \ref{anticode} applied to $\mC^\perp$ gives $\dim_{\F_q}(\mC^\perp) \le m \cdot \mbox{maxrk}(\mC^\perp)$, i.e., 
$\dim_{\F_q}(\mC) \ge m(k-\mbox{maxrk}(\mC^\perp))$. The same proposition applied to $\mC$ gives
$\dim_{\F_q}(\mC) \le m \cdot \mbox{maxrk}(\mC)$. Hence we have
\begin{equation}\label{ineq}
 m(k-\mbox{maxrk}(\mC^\perp)) \le \dim_{\F_q}(\mC) \le m \cdot \mbox{maxrk}(\mC).
\end{equation}
In particular, $k-\mbox{maxrk}(\mC^\perp) \le  \mbox{maxrk}(\mC)$. Given
the inequalities in (\ref{ineq}), it is easy to see that the bound is attained
if and only if both $\mC$ and $\mC^\perp$ are optimal anticodes, which occurs precisely when $\mC$ is an optimal anticode by Theorem \ref{dualanticode}.
 \end{proof}

\section{Matrices with given rank and $h$-trace} \label{sec5}

In this section we apply Corollary \ref{formule}, i.e., the MacWilliams identities for Delsarte codes,
to  classical problems in
enumerative combinatorics, deriving a recursive formula for
the number of $k\times m$ matrices over $\F_q$ with prescribed rank
and $h$-trace.

\begin{definition}
 Let $M \in \mbox{Mat}(k\times m, \F_q)$, and let $1 \le h \le \min \{ k,m\}$
be an integer. The \textbf{$h$-trace} of $M$ is defined by
$$\mbox{Tr}_h(M):= \sum_{i=1}^h M_{ii}.$$
\end{definition}

\begin{remark}
 Since for any matrix $M$ we have $\mbox{Tr}_h(M)=\mbox{Tr}_h(M^t)$, without
loss of generality in the following we only treat the case $k \le m$. Notice also
that when $k=m$ we have $\mbox{Tr}_k(M)=\mbox{Tr}(M)$. Hence the $h$-trace generalizes the trace of a matrix.
\end{remark}

\begin{notation} \label{not}
 Given integers $1 \le k \le m$, $0 \le r \le k$ and $1 \le h \le
k$, we denote by $n_q(k\times m, r,h)$ the number of matrices $M \in
\mbox{Mat}(k\times m, \F_q)$ such that $\mbox{rk}(M)=r$ and $\mbox{Tr}_h(M)=0$.
We also denote by $n_q(k\times m, r,0)$ the number of matrices in
$\mbox{Mat}(k\times m, \F_q)$ of rank $r$.
\end{notation}

\begin{lemma}\label{conto}
 Let $1 \le k \le m$ and $0 \le r \le k$ be integers. We have
 $$n_q(k\times m, r,0)=\begin{bmatrix} m \\ r \end{bmatrix}
\cdot \prod_{i=0}^{r-1} (q^k-q^i).$$ 
\end{lemma}

\begin{proof}[Sketch of proof]
For a given vector subspace $U \subseteq \F_q^m$ with $\dim_{\F_q}(U)=r$, the
number of matrices $M \in \mbox{Mat}(k\times m, \F_q)$ whose row space equals
$U$ is precisely the number of full-rank $r \times k$ matrices,
which is $\prod_{i=0}^{r-1} (q^k-q^i)$. The thesis follows from the fact that
the number of subspaces  $U \subseteq \F_q^m$ with $\dim_{\F_q}(U)=r$ is
$\begin{bmatrix} m \\ r \end{bmatrix}$.
\end{proof}

\begin{remark}
 We notice that if one has the number of matrices in $\mbox{Mat}(k\times m,
\F_q)$ of rank $r$ and zero $h$-trace, then he can also determine the number of matrices
in $\mbox{Mat}(k\times m, \F_q)$ of rank $r$ and $h$-trace equal to $\alpha$,
for any $\alpha$ in $\F_q$. Since the number of $k\times m$ matrices over $\F_q$
of rank $r$ is given by Lemma \ref{conto}, this fact is trivial
when $q=2$. On the other hand, if $q>2$ and $\alpha \neq \beta$ are non-zero
elements of $\F_q$, then the map $\mbox{Mat}(k\times m,
\F_q) \to \mbox{Mat}(k\times m, \F_q)$ defined by $M \mapsto \alpha^{-1}\beta M$
gives a bijection between the rank $r$ matrices with $h$-trace equal to $\alpha$
and the rank $r$ matrices with $h$-trace equal to
$\beta$. It follows that for any $\alpha \in \F_q \setminus \{ 0 \}$ the
number of matrices in
$\mbox{Mat}(k\times m, \F_q)$ with rank $r$ and $h$-trace equal to $\alpha$ is
$$\frac{n_q(k\times m, r,0)-n_q(k\times m, r,h)}{q-1},$$
where $n_q(k\times m, r,0)$ is explicitly given by Lemma \ref{conto}.
\end{remark}

\begin{remark}\label{refconto}
 The usual way of computing $n_q(k\times k,k,k)$ involves the Bruhat
decomposition of $\mbox{GL}_k(\F_q)$ and the theory of $q$-analogues (see  \cite{stanley}, Proposition 1.10.15). A different 
approach proposed in \cite{li} is based on Gauss sums over finite fields and properties of the Borel
subgroup of $\mbox{GL}_k(\F_q)$. In \cite{buck} Buckheister derived a
recursive description for $n_q(k\times k,r,k)$ using an elementary argument, and
in \cite{bender} Bender applied the results of
\cite{buck} to provide a closed formula for $n_q(k\times k,r,k)$. As Stanley
observed (\cite{stanley}, page 100), the description of \cite{buck} is quite
complicated.  The
following Theorem \ref{htraccia} provides a new recursive formula
for the numbers $n_q(k\times m,r,h)$ which easily follows from Corollary
\ref{formule}. An explicit version of the same formula  can be found in 
 Appendix \ref{appa}.
\end{remark}

\begin{theorem} \label{htraccia}
 Let $1 \le k \le m$ and $1 \le h \le
k$ be integers. For all $0 \le r \le k$ the numbers
$n_q(r,h):=n_q(k\times m, r,h)$ are
recursively computed by the following
formulas.
$$n_q(r,h)= \left\{ \begin{array}{ll} 
1 & \mbox{ if \ $r=0$,} \\ 
q^{mr-1} \left( \begin{bmatrix} k \\
r \end{bmatrix} +(q-1)  \begin{bmatrix} k-h \\
r \end{bmatrix} \right) - \mathlarger{\sum}_{j=0}^{r-1} {n_q(j,h)
\begin{bmatrix} k-j \\ r-j \end{bmatrix}} & \mbox{ if \ $1 \le r \le k-h$,}  \\
q^{mr-1} \begin{bmatrix} k \\
r \end{bmatrix} - \mathlarger{\sum}_{j=0}^{r-1} n_q(j,h)
\begin{bmatrix}
k-j \\ r-j \end{bmatrix} & \mbox{ if \ $k-h+1 \le r \le k$.}
\end{array}
 \right.\ $$
\end{theorem}

\begin{proof}
 We fix $1 \le h \le k$.  Let $M \in \mbox{Mat}(k\times m, \F_q)$ be the matrix
defined by
 $$M_{ij}:= \left\{ \begin{array}{ll} 1 & \mbox{ if  $i=j \le h$,} \\ 0 &
\mbox{ otherwise.} \end{array}\right.\ $$
Let $\mC:= \langle M \rangle \subseteq \mbox{Mat}(k\times m, \F_q)$ be the
Delsarte code
generated by $M$ over $\F_q$. It
is easy to check that for any matrix $N \in \mbox{Mat}(k\times m, \F_q)$  we
have $\mbox{Tr}_h(N)=\mbox{Tr}(M
N^t) = \langle M,N\rangle$. As a consequence, the set of
matrices in $\mbox{Mat}(k \times m, \F_q)$ with zero $h$-trace is
precisely $\mC^\perp$. Hence, we have
$n_q(r,h)=B_r$ for all $0 \le r \le k$, where  ${(B_j)}_{j \in \N}$ is the rank
distribution of $\mC^\perp$. If ${(A_i)}_{i \in \N}$ denotes
the rank distribution
of $\mC$, then we clearly have $A_0=1$, $A_h=q-1$, and $A_i=0$ for $i \notin \{
0,h\}$. 
The theorem now follows from Corollary \ref{formule}.
\end{proof}

\begin{example}
 Let $q=4$, $k=3$, $m=4$. Theorem \ref{formule} allows us to compute all
 the values of $n_4(3\times 4,r,h)$ as in Table \ref{tab}.

\begin{table}[h!]
\centering
  \begin{tabular}{c||c|c|c|c|}
          & $r=0$ & $r=1$ & $r=2$ & $r=3$  \\
          \hline   \hline 
  $h=1$   &   1   &   2283    & 381780   &  3810240  \\ \hline
  $h=2$   &   1   &   1515    &   336468  &  3856320 \\ \hline
  $h=3$   & 1     &  132     &  337428  &  3855552  \\ \hline
 \end{tabular}
 \caption{Values of $n_4(3\times 4,r,h)$.}
 \label{tab}
\end{table} 
\end{example}

\section*{Conclusions}
In this paper we prove that the duality theory of linear Delsarte codes
generalizes
the duality theory of linear Gabidulin codes. The relation between the two
duality theories is
described through trace-orthogonal bases of finite fields. We also give an
elementary proof of MacWilliams
identities for the general case of Delsarte codes, and show how to employ them
to re-establish in a very concise way
the main results of the theory of rank-metric codes. This also proves that
MacWilliams identities may be taken as a starting point for the theory of
rank-metric codes. We also investigate optimal Delsarte anticodes, and
characterize them in terms of MRD codes. Finally, we show an application of our
results solving a problem in enumerative combinatorics in an elementary way.

\section*{Acknowledgement}
The author is grateful to Elisa Gorla and to the Referees for many useful
suggestions that
improved the presentation of the paper.

\appendix

\section{Explicit form of Theorem \ref{theo} and \ref{htraccia}} \label{appa}

Using known properties of binomial coefficients one can show that 
Theorem \ref{theo} implies  Theorem 3.3 of \cite{del1} as an easy
corollary.
The following result, first proved by Delsarte using the theory of association
schemes,
 may be regarded as the explicit version of Theorem \ref{theo}.

\begin{theorem} \label{eq_del}
 Let $\mathcal{C} \subseteq \mbox{Mat}(k \times m, \F_q)$ be a code.
 Let
${(A_i)}_{i \in \N}$ and ${(B_j)}_{j \in \N}$ be the rank distributions of $\mC$
and $\mC^\perp$,
respectively. We have
$$B_j=\frac{1}{|\mC|} \mathlarger{‎‎\sum}_{i=0}^k A_i 
{\mathlarger{‎‎\sum}_{s=0}^k {(-1)}^{j-s}} q^{ms+ \binom{j-s}{2}}
\begin{bmatrix} k-s \\ k-j  \end{bmatrix} \begin{bmatrix} k-i \\ s 
\end{bmatrix}$$
for $j =0,...,k$.
\end{theorem}

\begin{proof} 
Throughout this proof the rows and columns of matrices are labeled from $0$ to
$k$ for convenience (instead of from $1$ to $k+1$).
Define the matrix $P \in \mbox{Mat}(k+1 \times k+1,\F_q)$ by
 $$P_{ji}:=\frac{1}{|\mC|} \sum_{s=0}^k (-1)^{j-s} q^{ms+ \binom{j-s}{2}} 
\begin{bmatrix} k-s \\ k-j  \end{bmatrix} \begin{bmatrix} k-i \\ s 
\end{bmatrix}$$
for $j,i \in \{0,...,k\}$. We can write the statement in matrix form as
$(B_0,...,B_k)^t= P \cdot (A_0,...,A_k)^t$. Define  matrices
$S,T \in \mbox{Mat}(k+1 \times k+1,\F_q)$ by
$$S_{ij}:= \begin{bmatrix} k-j \\ i-j  \end{bmatrix}, \ \ \ \ \ 
T_{ij}:= \frac{q^{mi}}{|\mC|} \begin{bmatrix} k-j \\ i  \end{bmatrix}$$
for $i,j \in \{0,...,k\}$. We notice that $S$ is invertible, since
it is lower-triangular and $S_{ii}=1$ for $i=0,...,k$.
Theorem \ref{theo} reads
$S \cdot (B_1,...,B_k)^t = T \cdot (A_0,...,A_k)^t$, i.e., 
$(B_1,...,B_k)^t = S^{-1}T \cdot (A_0,...,A_k)^t$. Hence it suffices to prove
$P=S^{-1}T$, i.e., $T=SP$. Fix arbitrary integers $i,j \in \{0,...,k\}$. We have
\begin{eqnarray*}
 (SP)_{ij} &=& \frac{1}{|\mC|}\sum_{r=0}^k \begin{bmatrix} k-r \\ i-r 
\end{bmatrix} 
\sum_{s=0}^k (-1)^{r-s} q^{ms+ \binom{r-s}{2}} 
\begin{bmatrix} k-s \\ k-r  \end{bmatrix} \begin{bmatrix} k-j \\ s 
\end{bmatrix} \\
&=&  \frac{1}{|\mC|} \sum_{s=0}^k  q^{ms} \begin{bmatrix} k-j \\ s 
\end{bmatrix} 
\sum_{r=0}^k  \begin{bmatrix} k-r \\ i-r  \end{bmatrix} 
 (-1)^{r-s} q^{\binom{r-s}{2}} 
\begin{bmatrix} k-s \\ k-r  \end{bmatrix}.
\end{eqnarray*}
 Clearly,
$$\begin{bmatrix} k-r \\ i-r  \end{bmatrix}= \begin{bmatrix} k-r \\ k-i 
\end{bmatrix},$$
and using the definition of Gaussian binomial coefficient one finds
$$\begin{bmatrix} k-s \\ k-r  \end{bmatrix} \begin{bmatrix} k-r \\ k-i 
\end{bmatrix} =
\begin{bmatrix} k-s \\ k-i \end{bmatrix} \begin{bmatrix} i-s \\ r-s 
\end{bmatrix}.$$
Hence we have
\begin{eqnarray*}
\sum_{r=0}^k  \begin{bmatrix} k-r \\ i-r  \end{bmatrix} 
 (-1)^{r-s} q^{\binom{r-s}{2}} 
\begin{bmatrix} k-s \\ k-r  \end{bmatrix} &=&
\sum_{r=0}^k  \begin{bmatrix} k-s \\ k-i \end{bmatrix} \begin{bmatrix} i-s \\
r-s  \end{bmatrix}
(-1)^{r-s} q^{\binom{r-s}{2}} \\
&=& \begin{bmatrix} k-s \\ k-i \end{bmatrix} \sum_{r=0}^k   \begin{bmatrix} i-s
\\ r-s  \end{bmatrix}
(-1)^{r-s} q^{\binom{r-s}{2}} \\
&=& \begin{bmatrix} k-s \\ k-i \end{bmatrix} \sum_{r=-s}^{k-s}  
 \begin{bmatrix} i-s \\ r  \end{bmatrix}
(-1)^{r} q^{\binom{r}{2}} \\
&=& \begin{bmatrix} k-s \\ k-i \end{bmatrix} \sum_{r=0}^{i-s}  
 \begin{bmatrix} i-s \\ r  \end{bmatrix}
(-1)^{r} q^{\binom{r}{2}} \\
&=& \left\{ \begin{array}{ll} 1 & \mbox{ if $s=i$,} \\ 0 & \mbox{ otherwise,}
\end{array} \right.\
\end{eqnarray*}
where the last equality follows from the $q$-Binomial Theorem (\cite{stanley},
page 74).
It follows
$$(SP)_{ij}= \frac{1}{|\mC|} q^{mi} \begin{bmatrix} k-j \\ i 
\end{bmatrix}=T_{ij},$$
as claimed.
\end{proof}

Arguing as in the proof of Theorem \ref{htraccia} and replacing Corollary
\ref{formule}
with Theorem \ref{eq_del} we easily obtain the following explicit version of
Theorem \ref{htraccia}.

\begin{theorem} \label{htrex}
 Let $1 \le k \le m$, $1 \le h \le
k$ and $0 \le r \le k$ be integers. The number of $k \times m$ matrices
over $\F_q$ having rank $r$ and zero $h$-trace is 
$$n_q(k\times m, r,h)=\frac{1}{q} \mathlarger{‎‎\sum_{s=0}^k} {(-1)}^{r-s}
q^{ms+ \binom{r-s}{2}}
\begin{bmatrix} k-s \\ k-r  \end{bmatrix} \left( \begin{bmatrix} k \\ s 
\end{bmatrix} +
(q-1)  \begin{bmatrix} k-h \\ s  \end{bmatrix}\right).$$ 
\end{theorem}

\begin{remark}
 Theorem \ref{htrex} generalizes the works cited in Remark \ref{refconto}.
\end{remark}

\end{document}